\documentclass[conference]{IEEEtran}

\usepackage{tikz}
\usepackage{cite}
\usepackage{amsmath,amssymb,amsfonts,amsthm}
\usepackage{algorithmic}
\usepackage{graphicx}
\usepackage{textcomp}
\usepackage{xcolor}

\usepackage{tikz}
\usetikzlibrary{arrows,shapes}
\usetikzlibrary{positioning}
\tikzstyle{block} = [draw,minimum size=2.5em]

\newcommand{\vo}[1]{\boldsymbol{#1}}
\newcommand{\set}[1]{\mathbb{#1}}
\newcommand{\real}{\set{R}}

\newcommand{\x}{\vo{x}}
\newcommand{\xd}{\dot{\vo{x}}}

\newcommand{\z}{\vo{z}}
\renewcommand{\u}{\vo{u}}
\newcommand{\uhat}{\vo{\hat{u}}}

\renewcommand{\d}{\vo{d}}
\newcommand{\dbar}{\vo{\bar{d}}}
\newcommand{\w}{\vo{w}}
\newcommand{\wbar}{\vo{\bar{w}}}
\newcommand{\kappab}{\vo{\kappa}}

\newcommand{\X}{\vo{X}}

\newcommand{\Y}{\vo{Y}}

\newcommand{\A}{\vo{A}}
\newcommand{\B}{\vo{B}}
\newcommand{\C}{\vo{C}}
\newcommand{\D}{\vo{D}}
\newcommand{\I}{\vo{I}}
\newcommand{\Q}{\vo{Q}}

\newcommand{\K}{\vo{K}}
\newcommand{\F}{\vo{F}}
\newcommand{\W}{\vo{W}}
\newcommand{\G}{\vo{G}}
\newcommand{\M}{\vo{M}}
\newcommand{\V}{\vo{V}}

\newcommand{\tr}[1]{\textbf{tr}\left[#1\right]}
\newcommand{\diag}[1]{\textbf{diag}\left(#1\right)}

\newcommand{\norm}[1]{\left\|#1\right\|}

\newcommand{\elab}[1]{\label{eqn:#1}}
\newcommand{\eqn}[1]{(\ref{eqn:#1})}

\newcommand{\flab}[1]{\label{fig:#1}}
\newcommand{\fig}[1]{Fig.\ref{fig:#1}}

\newtheorem{theorem}{Theorem}
\newtheorem{remark}{Remark}

\begin{document}
\title{Quantifying Maximum Actuator Degradation for a Given $\mathcal{H}_2/\mathcal{H}_\infty$ Performance with Full-State Feedback Control}

\author{\IEEEauthorblockN{\bf Hrishav Das}
\IEEEauthorblockA{Aerospace Engineering\\
Indian Institute of Technology Madras\\
Chennai, India\\
ae21b023@smail.iitm.ac.in}
\and
\IEEEauthorblockN{\bf Eliot Nychka}
\IEEEauthorblockA{Aerospace Engineering\\
Texas A\&M University\\
College Station, TX, USA\\
eliot.nychka@tamu.edu}
\and
\IEEEauthorblockN{\bf Raktim Bhattacharya}
\IEEEauthorblockA{Aerospace Engineering\\
Texas A\&M University\\
College Station, TX, USA\\
raktim@tamu.edu}
}

\maketitle

\begin{abstract}
    In this paper, we address the issue of quantifying maximum actuator degradation in linear time invariant dynamical systems. We present a new unified framework for computing the state-feedback controller gain that meets a user-defined closed-loop performance criterion while also maximizing actuator degradation. This degradation is modeled as a first-order filter with additive noise. Our approach involves two novel convex optimization formulations that concurrently determine the controller gain, maximize actuator degradation, and maintain the desired closed-loop performance in both the $\mathcal{H}_2$ and $\mathcal{H}_\infty$ system norms. The results are limited to open-loop stable systems. We demonstrate the application of our results through the design of a full-state feedback controller for a model representing the longitudinal motion of the F-16 aircraft.\end{abstract}
\IEEEpeerreviewmaketitle

\section{Introduction}
This paper explores quantifying maximum actuator degradation in linear dynamical systems, a critical issue in engineering control systems. Actuator degradation results from wear, environmental influences, or aging components, leading to diminished performance and, in severe cases, system failure. In aerospace engineering, for example, the health of control systems is vital for the safe and efficient operation of aircraft; actuator degradation could significantly compromise flight control, endangering safety. Similar issues also exist in other engineering fields and become essential in safety-critical applications. Hence, developing methodologies for assessing and quantifying maximum actuator degradation is critical for quantifying robustness.

\subsection{Literature Survey}
In the current state of the art, there are generally two approaches to addressing faults or degradation in systems and their controllers: enhancing the robustness and fault tolerance of controllers by integrating degradation estimates, or developing improved models for actuator degradation.

Papers such as \cite{khelassi2011fault, zhang2022solving, rezaeizadeh2023reliability,shen2016robust,4304100} look at methods related to enhancing controller robustness and fault tolerance by incorporating degradation estimates. Specifically, \cite{khelassi2011fault} presents a fault-tolerant control design that utilizes actuator health information to improve system reliability and reduce reliance on more sensitive actuators. The approach involves solving a fault-tolerant tracking problem using a Linear Matrix Inequality (LMI) framework while taking into account actuator health. Zhang et.al. \cite{zhang2022solving} introduces an algorithm that models actuator degradation using the Weiner Diffusion method. This degradation is then incorporated as a separate state variable in the so called Degradation Informed State System (DISS) within a controller. The objective is to optimize tracking performance while minimizing degradation. The controller tracks a reference trajectory of physical states such as position, orientation, and degradation quantity, employing a Model Predictive Control (MPC) approach, although other optimization-based methods could also be used. Notably, the paper also addresses the concept of `maximum degradation', which is stochastically computed from historical data on actuators and their health to impose constraints on the MPC. In \cite{rezaeizadeh2023reliability}, the primary objective is to introduce a controller for electric vehicles that enhances speed tracking, reliability, and lifespan. The approach is applied to a power converter that operates a permanent magnet synchronous motor (PMSM). The paper presents a method that employs an extended $\mathcal{H}_{\infty}$ design framework alongside reliability models to mitigate thermal stress on the power electronic Insulated Gate Bipolar Transistors (IGBTs). The results demonstrate that the proposed control strategies are more effective in real-world scenarios compared to existing state-of-the-art methods. Shen et. al. \cite{shen2016robust} presents a high-level Model Reference Adaptive Control (MRAC) approach to generate desired torques for tracking specific trajectories and orientations while addressing disturbances. These torques serve as input to a robust control allocation scheme that apportions control effort, taking into account practical constraints on each control input and inaccuracies in fault estimation. This study addresses uncertainties in actuator effectiveness and fault detection and diagnostics, without focusing on estimating the fault or calculating the maximum allowable fault. In \cite{4304100}, a method to design a robust and optimal controller for open-loop unstable systems with actuator redundancy is proposed. Primarily the system outlines the LQR method as a `robust Linear Quadratic (LQ) regulator' and shows it to be equivalent to a traditional LQ regulator. 

A considerable amount of research focuses on developing models for actuator degradation. For instance, \cite{wang2015fault, zhang2003fault} models actuator degradation by adjusting its influence (modifying the B matrix) and discusses a method for estimating this degradation. In  \cite{zhang2008incorporating} various degradation models are explored that can be selected based on the actuator's degradation state. It also presents a technique for adjusting the desired state input to the controller to prevent saturation or exceeding limits in degraded actuators.

Nguyen et al. \cite{nguyen2014feedback} and Si et al. \cite{8902212} incorporate the concept of Remaining Useful Life (RUL) in their degradation modeling formulations. The work of Nguyem et.al. employs a statistical approach, considering discrete time impacts where each impact contributes to degradation within a probabilistic range. The RUL of actuators is determined by simulating the state until a specified precision criterion is no longer met (indicating failure). The work of Si et. al. integrates the degradation process model of the actuator with the system's state transition model and utilizes a particle filter algorithm to estimate both the hidden degradation state and the system state.

There has been limited work on quantifying worst-case or maximum degradation of actuators. Few papers \cite{nguyen2014feedback, 8902212, zhang2022solving} come closest to addressing the broader issue of computing maximum actuator degradation for an acceptable closed-loop performance. There has been no prior work, to the best of our knowledge, that computes the maximum actuator degradation for a given closed-loop performance.

\subsection{Key Contributions}
We introduce a novel unified framework to compute the state-feedback controller gain that satisfies a user-specified closed-loop performance criterion while maximizing actuator degradation. The degradation is represented as a first-order filter with additive noise. We propose two new convex optimization formulations that simultaneously determine the controller gain, maximize actuator degradation, and ensure the desired closed-loop performance in both the $\mathcal{H}_2$ and $\mathcal{H}_\infty$ system norms. The key contributions are detailed in theorems 1 and 2. This is the first time a joint formulation is presented in the $\mathcal{H}_2/\mathcal{H}_\infty$ framework. However, these results are limited to open-loop stable systems.

\section{Preliminaries}
The aircraft can generally experience actuator faults that can be classified as either \textit{total} or \textit{partial} faults. In actuators, a total fault is a complete loss of control action, and examples include a stuck or floating actuator, which produces no controllable actuation. They can occur due to breakage, cut or burnt wiring, short-circuit presence of a foreign body in the actuator, etc. Partially failed actuators produce only a part of the normal actuation and correspond to reduced rate and magnitude limits. These faults can result from hydraulic or pneumatic leakage, increased resistance, a fall in the supply voltage, etc.

In this paper, complete and partial faults are modeled as shown in \fig{faultModel}, where a loss in the actuator rate is modeled using a low-pass filter, $F(s):= \frac{\gamma_u}{s/\omega_c+1}$, with cutoff frequency $\omega_c\ge 0$. The biases and inaccuracies are modeled using an additive noise $w_a(t)$ and is assumed to be a norm-bounded signal $\|w_a(t)\|_2 \le \gamma_a$. The control magnitude loss is modeled as a constraint on $\gamma_u$ the dc gain of the actuator. Therefore, with parameters $\omega_c,\gamma_a$ and $\gamma_u$, we can capture a rich class of actuator faults commonly encountered in control systems. 

\begin{figure}[h!]
\centering
\begin{tikzpicture}[>=latex',thick]
    \node [block] (filter) {\begin{tabular}{c}Fault  \\ Filter $F(s)$ \end{tabular}};
    \node (ya) [draw=black,right=1 cm of filter, node distance=1cm,circle, inner sep=0.01ex] {$+$};
 
	\node[left=0.5cm of filter]  (uc){\begin{tabular}{c}Control \\Signal  \\ $u(t)$ \end{tabular} };
	\node[right=0.5cm of ya]  (ur){\begin{tabular}{c}Faulty \\ Control \\ Signal \\ $\hat{u}(t)$ \end{tabular}};
	\node (wa) [above=3ex of ya, node distance=1cm] {\begin{tabular}{c}Actuator \\Noise  \\ $w_a(t)$ \end{tabular}};
    \draw[->] (uc) -- (filter);
	\draw[->] (filter) -- (ya) node[midway,below]{$x_F(t)$};
	\draw[->] (ya) -- (ur);
	\draw[->] (wa) -- (ya);
\end{tikzpicture}
\caption{Modeling of a faulty actuator.}
\flab{faultModel}
\end{figure}
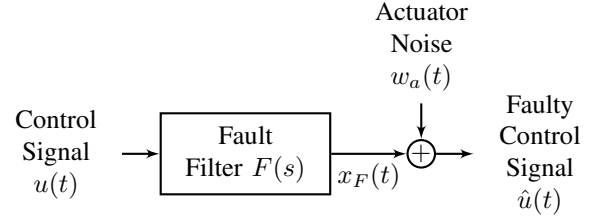

For  systems with multiple inputs, the dynamics of the filter can be written as
\begin{subequations}
\begin{align}
\xd_F(t) &= -\diag{\vo{\omega}_c}\x_F(t) + \diag{\vo{\omega}_c}\u(t),\\
\hat{\u}(t) &= \x_F(t) + \w_a(t), \elab{act_fault}
\end{align}
\elab{filter_dynamics}
\end{subequations}
where $n_u$ dimensional vectors $\x_F$, $\u$, $\w_a$, and $\hat{\u}$ are the filter states, control signals, actuator noises, and faulty control signals, respectively. 

We consider the open-loop LTI system,
\begin{subequations}
\begin{align}
    \xd(t) &=\A\x(t) + \B_u\u(t) + \B_d\d(t),\elab{olp_dyn}\\
    \z(t) &=\C_z\x(t), \elab{zdef}
\end{align}
\elab{olp_sys}
\end{subequations}
where $\x\in\real^{n_x}$ is the state vector, $\u\in\real^{n_u}$ is the control vector, $\d\in\real^{n_d}$ is the  exogeneous disturbance signal, and $\z\in\real^{n_z}$ is the \textit{controlled} variable. 

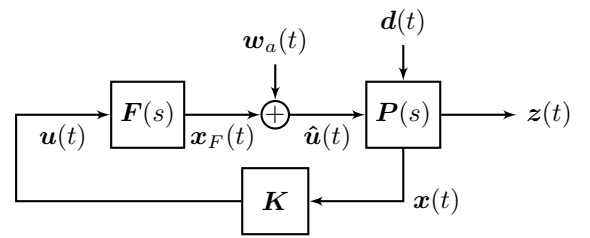
\begin{figure}[h!]
\centering
\begin{tikzpicture}[>=latex',thick,scale=1,transform shape]
    \node [block] (filter) {$\F(s)$};
    \node (ya) [draw=black,right=1 cm of filter, node distance=1cm,circle, inner sep=0.01ex] {$+$};
 
	\coordinate[left=1.25cm of filter] (uc){};
	
	\node[right=1cm of ya,block] (P) {$\vo{P}(s)$};
	\node[below=.5cm of ya,block] (K) {$\K$};

	\node[right=1cm of P] (z) {$\z(t)$};

	\node (wa) [above=3ex of ya, node distance=1cm] {$\w_a(t)$};
	
    \node (d) [above=3ex of P, node distance=1cm] {$\d(t)$};

    \draw[->] (K) -| (uc) -- (filter) node[midway,below]{$\u(t)$};     
    \draw[->] (filter) -- (ya) node[midway, below]{$\x_F(t)$};
    \draw[->] (wa) -- (ya);
    \draw[->] (ya) -- (P) node[midway,below]{$\uhat(t)$};
    \draw[->] (d) -- (P);
    \draw[->] (P) |- (K) node[midway,right]{$\x(t)$};
    \draw[->] (P) -- (z);
\end{tikzpicture}
\caption{Closed-loop system with full-state feedback controller and faulty actuator.}
\flab{faultyCLP}
\end{figure}

This paper considers a full-state feedback formulation, i.e., $\u := \K\x$. Therefore, the closed-loop system with the faulty actuators is shown in \fig{faultyCLP}, and its dynamics is given by
\begin{subequations}
\begin{multline}
\begin{bmatrix} \xd\\ \xd_F \end{bmatrix} = 
\begin{bmatrix}
\A & \B_u \\ \underbrace{\diag{\vo{\omega}_c}\K}_{:=\V} & -\diag{\vo{\omega}_c}
\end{bmatrix}\begin{bmatrix} \x\\ \x_F \end{bmatrix} \\ + \begin{bmatrix}\B_d & \B_u \\ \vo{0} & \vo{0}\end{bmatrix}\begin{bmatrix}\d \\ \w_a \end{bmatrix},
\end{multline}
\begin{align}
\z &= \begin{bmatrix} \C_z & \vo{0} \end{bmatrix}\begin{bmatrix} \x\\ \x_F \end{bmatrix}.
\end{align}
\elab{system_dynamics}
\end{subequations}

The problem is determining $\K$ that bounds the $\mathcal{H}_2/\mathcal{H}_\infty$ norm of the closed-loop transfer function from $\begin{bmatrix}\d\\ \w_a\end{bmatrix} \to\z$ while maximizing actuator faults. In this paper, we propose that the actuator fault is maximized when $\|\w_a(t)\|_2$ is maximized, and the actuation rate and the actuator's DC gain are minimized. The latter two are achieved by minimizing $\|\vo{\omega}_c\|_2$ and $\tr{\V^T\V}$.

Thus, the problem is reduced to designing a control law that meets the specified closed-loop performance with the maximum actuator noise and minimum control rate and magnitude. 

\begin{remark}
    In \eqn{zdef}, we assume that $\z(t)$ is not dependent on $\u(t)$, which is necessary to solve problem in a convex optimization framework. Our future work will explore formulations to remove this restriction.
\end{remark}

\begin{remark}
    The results presented in the paper apply to open-loop stable systems, due to the assumption of a particular structure for the Lyapunov function, which results in a convex formulation.
\end{remark}

\section{Technical Results}
We will use the following standard notations: 
\begin{itemize}
\item $\real^n$ is the space of $n$-dimensional vectors,
\item $\real_+^n$ is the space of $n$-dimensional vectors with real non-negative entries,
\item $\real^{m\times n}$ is the space of $m\times n$ dimensional vectors with real entries,
\item $\mathbb{S}^n_+$ is the space of $n\times n$ symmetric positive semi-definite matrices with real entries,
\item vectors and matrices are represented in bold font, and
\item scalars are represented in normal font.
\end{itemize}
\begin{theorem}
    For the system in \eqn{system_dynamics}, the state-feedback gain that guarantees $\norm{\begin{bmatrix}\d\\ \w_a\end{bmatrix} \to\z}_\infty \leq \gamma$ for the system in \eqn{system_dynamics}, with maximum actuator noise, and minimum control magnitude and rate, is given by the following convex optimization problem,
    \begin{align*}
& \min_{\kappab_a, \vo{\omega}_c, \gamma_{\x_F}, \V,\Y,\Q} \lambda_a \norm{\kappab_a}_2 + \lambda_{\vo{\omega}_c}\norm{\vo{\omega}_c}_2 + \lambda_{\x_F}\gamma_{\x_F},\\
& \text{ subject to} \notag\\
&   \begin{bmatrix}
    \vo{P} + \vo{P}^T &  \begin{bmatrix} \Y\B_d & \Y\B_u\\ \vo{0} & \vo{0}\end{bmatrix} & \begin{bmatrix}\C_z^T \\ \vo{0}\end{bmatrix}\\[2mm]
        \ast & \gamma\begin{bmatrix}\W_d^{-2} & \vo{0}\\ \vo{0} & \diag{\kappab_a}\end{bmatrix} & \begin{bmatrix}\vo{0}\\ \vo{0}\end{bmatrix} \\
        \ast & \ast & \gamma\I
\end{bmatrix} \leq 0,\\
&  \begin{bmatrix}\Q & \V^T \\
    \V & \I 
\end{bmatrix} \geq 0, \tr{\Q} \leq \gamma_{\x_F},  \\
& \kappab \in \real_+^{n_u}, \vo{\omega}_c \in\real_+^{n_u}, \gamma_{\x_F} \in\real_+, \\
&\V \in\real^{n_u\times n_x},\Y \in \mathbb{S}_+^{n_x}, \Q\in \mathbb{S}_+^{n_z}, 
\end{align*}
with 
$$
\vo{P} := \begin{bmatrix} \Y\A &  \Y\B_u\\
    \V & -\diag{\vo{\omega}_c} \end{bmatrix}, 
$$
where $\lambda_a$, $\lambda_{\vo{\omega}_c}$, and $\lambda_{\x_F}$ are user-defined weights, and $\W_d$ is a given scaling matrix for $\d(t)$.  The controller gain can be recovered as $\K := \diag{\vo{\omega}_c}^{-1}\V$.
\end{theorem}
    
    \begin{proof}
        We introduce diagonal matrix weights $\W_a$ and $\W_d$ such that 
        \begin{align}
            \w_a(t) := \W_a\wbar_a(t), && \d(t) := \W_d\dbar(t),
        \end{align}
        where $\norm{\wbar(t)}_2 = \norm{\dbar(t)}_2 = 1$, $\W_d$ is known, and $\W_a$ is unknown. Therefore, the closed-loop system in terms of the normalized exogeneous signals is given by
        
        \begin{align}
           &\left[\begin{array}{c}\xd\\ \xd_F \\ \hline \z \end{array}\right] 
           = \notag \\
           &\left[\begin{array}{cc|cc}
            \A & \B_u & \B_d\W_d & \B_u\W_a\\
            \diag{\vo{\omega}_c}\K & -\diag{\vo{\omega}_c} & \vo{0} & \vo{0}\\\hline
            \C_z & \vo{0} & \vo{0} & \vo{0}
           \end{array}\right]\left[\begin{array}{c}
            \x\\ \x_F \\ \hline \dbar \\ \wbar_a  \end{array}\right].\elab{tf}
        \end{align}
    Suppose $\G(s)$ is the transfer function from $\begin{bmatrix}\dbar \\ \wbar_a \end{bmatrix} \to \z$
     as defined by \eqn{tf}, then the Bounded Real Lemma \cite{yaku62, yakubovich1967method,iwasaki2005generalized,rantzer1996kalman} states $\norm{\G(s)}_\infty \leq \gamma$ is equivalent to the inequality
    \begin{align}
    \begin{bmatrix}
        \M_{11}^T\X + \X\M_{11} & \X\M_{12} & \M_{21}^T \\
        \ast & -\gamma\I & \M_{22}^T \\
        \ast & \ast & -\gamma\I 
    \end{bmatrix} \leq 0,
    \elab{mi}
    \end{align}
    for $\X \geq 0$, where  $\M_{ij}$ are defined from the matrix partitions in \eqn{tf}. The matrix inequality in \eqn{mi}, is not convex due to the products of the unknowns $\vo{\omega}_c$, $\K$, $\W_a$, and $\X$.
    
    Introducing $\V := \diag{\vo{\omega}_c}\K$ and assuming a special structure for $\X$, 
    \begin{equation}
    \X := \begin{bmatrix} \Y & \vo{0} \\ \vo{0} & \I \end{bmatrix},\elab{Xdef}
    \end{equation}
   then
    $$
    \X\M_{11} = \begin{bmatrix} \Y\A &  \Y\B_u\\
        \V & -\diag{\vo{\omega}_c} \end{bmatrix}, 
    $$
    which is linear in $\Y$, $\vo{\omega}_c$, and $\V$. The controller gain can be recovered as $\K := \diag{\vo{\omega}_c}^{-1}\V$.
    
    The term $\X\M_{12}$ can be written as
    $$\X\M_{12} = \begin{bmatrix} \Y\B_d & \Y\B_u\\ \vo{0} & \vo{0}\end{bmatrix}\diag{\W_d,\W_a}.
    $$
    Therefore, the quadratic term corresponding to $\X\M_{12}$ in the Schur complement of \eqn{mi} can be expressed as $\begin{bmatrix} \Y\B_d & \Y\B_u\\ \vo{0} & \vo{0}\end{bmatrix}\diag{\W^2_d,\W^2_a} \begin{bmatrix} \Y\B_d & \Y\B_u\\ \vo{0} & \vo{0}\end{bmatrix}^T$. Defining,
    \begin{align*}
        \W^2_a &:= \diag{\kappab_a}^{-1},
    \end{align*}
    the matrix inequality in \eqn{mi} can be written as
  $$
    \begin{bmatrix}
        \vo{P} + \vo{P}^T &  \begin{bmatrix} \Y\B_d & \Y\B_u\\ \vo{0} & \vo{0}\end{bmatrix} & \begin{bmatrix}\C_z^T \\ \vo{0}\end{bmatrix}\\[2mm]
            \ast & -\gamma\begin{bmatrix}\W_d^{-2} & \vo{0}\\ \vo{0} & \diag{\kappab_a}\end{bmatrix} & \begin{bmatrix} \vo{0}\\ \vo{0}\end{bmatrix} \\
            \ast & \ast & -\gamma\I
    \end{bmatrix} \leq 0,
$$    
which is convex in $\Y,\V,\vo{\omega}_c$, and $\kappab_a$, where 
$$
\vo{P} := \begin{bmatrix} \Y\A &  \Y\B_u\\
    \V & -\diag{\vo{\omega}_c} \end{bmatrix}. 
$$

The DC gain matrix of the actuator is minimized, by minimizing a upper bound of $\tr{\V^T\V}$, i.e., minimize $\gamma_{\x_F}$ such that $\V^T\V \leq \Q$ and $\tr{\Q} \leq \gamma_{\x_F}$ and $\Q \geq 0$. The constraint on $\V$ can be expressed as the following LMI using Schur complement,
$$
\begin{bmatrix}
    \Q & \V^T \\ \V & \I
\end{bmatrix} \geq 0.
$$

The actuator noise is maximized by minimizing $\norm{\kappab_a}_2$, the control rate is minimized by minimizing $\norm{\vo{\omega}_c}_2$, and the control magnitude is minimized by minimizing $\gamma_{\x_F}$. Therefore, the cost function for the optimization is
    $$
    \min_{\kappab_a, \vo{\omega}_c, \gamma_{\x_F}, \V,\Y\geq0,\Q\geq0} \lambda_a \norm{\kappab_a}_2 + \lambda_{\vo{\omega}_c}\norm{\vo{\omega}_c}_2 + \lambda_{\x_F}\gamma_{\x_F},
    $$
    where $\lambda_a$, $\lambda_{\vo{\omega}_c}$, and $\lambda_{\x_F}$ are user-defined weights.
    \end{proof}
    
    We next present the result for $\mathcal{H}_2$ optimal closed-loop performance.
\begin{theorem}
    For the system in \eqn{system_dynamics} with $\D_d = 0$, the state-feedback gain that guarantees $\norm{\begin{bmatrix}\d\\ \w_a\end{bmatrix} \to\z}_2 \leq \gamma$ for a given $\gamma \geq 0$, with maximum actuator noise, and minimum control magnitude and rate, is given by the following convex optimization problem,

    \begin{align*}
        &\min_{\kappab_a, \vo{\omega}_c, \gamma_{\x_F}, \V,\Y,\Q_1,\Q_2} \lambda_a \norm{\kappab_a}_2 + \lambda_{\vo{\omega}_c}\norm{\vo{\omega}_c}_p + \lambda_{\x_F}\gamma_{\x_F},\\
        & \text{ subject to} \notag\\
        &\begin{bmatrix}
            \vo{P} + \vo{P}^T & \begin{bmatrix} \Y\B_d & \Y\B_u\\ \vo{0} & \vo{0} \end{bmatrix} \\
            \ast & -\begin{bmatrix}\W_d^{-2} & \vo{0}\\ \vo{0} & \diag{\kappab_a}\end{bmatrix}
        \end{bmatrix} \leq 0\\[2mm]
        &\begin{bmatrix}
            \Q_1 & \C_z & \vo{0}\\
            \ast & \Y & 0 \\
            \ast & \ast & \I
        \end{bmatrix} \geq 0, \;\begin{bmatrix}\Q_2 & \V^T \\
            \V & \I 
        \end{bmatrix} \geq 0,\\
        & \tr{\Q_1} \leq \gamma, \tr{\Q_2} \leq \gamma_{\x_F}, \kappab \in \real_+^{n_u}, \vo{\omega}_c \in\real_+^{n_u}, \gamma_{\x_F} \in\real_+, \\
        &\V \in\real^{n_u\times n_x},\Y \in \mathbb{S}_+^{n_x}, \Q_1\in \mathbb{S}_+^{n_z}, \Q_2\in \mathbb{S}_+^{n_x}, 
        \end{align*}
        with 
        $$
        \vo{P} := \begin{bmatrix} \Y\A &  \Y\B_u\\
        \V & -\diag{\vo{\omega}_c} \end{bmatrix}, 
        $$
        where $\lambda_a\geq0$, $\lambda_{\vo{\omega}_c}\geq0$, and $\lambda_{\x_F}\geq0$ are user-defined weights, and $\W_d$ is a given scaling matrix for $\d(t)$.  The controller gain can be recovered as $\K := \diag{\vo{\omega}_c}^{-1}\V$. 
        \end{theorem}

        \begin{proof}
            The condition for bounding the $\mathcal{H}_2$ norm from $\begin{bmatrix}\dbar\\ \wbar_a\end{bmatrix} \to \z$ is given by $\X\M_{11} + \M^T_{11}\X + \X\M_{12}\begin{bmatrix}\W_d^2 &0 \\ 0 &\W^2_a\end{bmatrix}\M^T_{12}\X \leq 0$, and $\tr{\M_{21}\X^{-1} \M^T_{21}}\leq \gamma$ \cite{feron1992numerical}, which for the problem considered are the following LMIs 
            \begin{align*}
            &\begin{bmatrix}
                \vo{P} + \vo{P}^T & \begin{bmatrix} \Y\B_d & \Y\B_u\\ \vo{0} & \vo{0} \end{bmatrix} \\
                \ast & -\begin{bmatrix}\W_d^{-2} & \vo{0}\\ \vo{0} & \diag{\kappab_a}\end{bmatrix}
            \end{bmatrix} \leq 0, \\[2mm]
            &\begin{bmatrix}
                \Q_1 & \C_z & \vo{0}\\
                \ast & \Y & 0 \\
                \ast & \ast & \I
            \end{bmatrix} \geq 0, \tr{\Q_1} \leq \gamma.
        \end{align*}
        Similar to theorem 1, $\tr{\V^T\V} \leq \gamma_{\x_F}$ is equivalent to $\begin{bmatrix}\Q_2 & \V^T\\ \V & \I \end{bmatrix} \geq 0$ and $\tr{\Q_2} \leq \gamma_{\x_F}$. Maximum actuator degradation is achieved by minimizing $\lambda_a \norm{\kappab_a}_2 + \lambda_{\vo{\omega}_c}\norm{\vo{\omega}_c}_2 + \lambda_{\x_F}\gamma_{\x_F}$, where $\lambda_a$, $\lambda_{\vo{\omega}_c}$, and $\lambda_{\x_F}$ are user-defined weights.
        \end{proof}

        \begin{remark}
            In both theorem 1 and 2, the minimum control rate is given by the elements $\vo{\omega}_c$, the  minimum control magnitude is characterized by $\gamma_{\x_F}$, and the maximum actuator noise is characaterized by the scaling $\diag{1/\sqrt{\kappab}}$, where the square-root and the inverse are defined to be element-wise. 
            
            Therefore, the optimal $\vo{\omega}_c$, $\gamma_{\x_F}$, and $\kappab$ quantities maximum actuator degradation (as defined above) for which a state-feedback controller is able to achieve the user specified-performance in terms of the closed-loop systems $\mathcal{H}_2$ or $\mathcal{H}_\infty$ norms.
        \end{remark}

\section{Examples}
\subsection{F-16 Flight Control Application}
We apply the results to the design of a full-state feedback controller for a model of the longitudinal motion of the F-16 aircraft. The model is based on the aerodynamics from the NASA Technical Report 1538 \cite{nguyen1979simulator}. The nonlinear model is trimmed at $10000$ $ft$ at $900$ $ft/s$ for steady-level flight. The longitudinal states are $\theta$ (pitch angle), $V_t$ (total velocity), $\alpha$ (angle-of-attack), $q$ (pitch rate); with trim value $\bar{\theta} = 5.95^\circ$, $\bar{V}_t = 900\; ft/s$, $\bar{\alpha} = 5.95^\circ$, $\bar{q} = 7.85 \; deg/s$. The control variables are $T$ (thrust), $\delta_e$ (elevator angle), and $\delta_{lef}$ (leading flap); with trim values $\bar{T} = 10461.84\; lb$, $\bar{\delta}_e = -3.82^\circ$, and $\delta_{lef} = 12.42^\circ$. We assume the aircraft experiences disturbance in the angle-of-attack $\alpha$, with $W_d=0.01$.

The open-loop system in \eqn{olp_sys} is given by the following system matrices

\begin{align*}
    &\A := \begin{bmatrix}
        0.0      &0.0        &0.0      &1.0\\
        -32.1699  &-0.0358  &-131.646   &-3.1099\\
          0.0     &-0.0002    &-1.5333  & 0.9281\\
          0.0     & 0.0003    &-4.6719  &-1.9076
    \end{bmatrix},\\
    &\B_u := \begin{bmatrix}
        0.0      &0.0      &0.0\\
        0.0016   &0.0525   &0.1574\\
       -0.0     &-0.0031   &0.0008\\
        0.0     &-0.4503  &-0.0614
    \end{bmatrix},\\
    &\B_d := \begin{bmatrix}
        0\\
        1\\
          0\\
          0
    \end{bmatrix}.
\end{align*}    

The control objective is to keep perturbations in the flight-path angle $\theta-\alpha$ and $V_t$ to be small due to the disturbance in $V_t$, in the $\mathcal{H}_2$ and $\mathcal{H}_\infty$ sense. Therefore,
\begin{alignat*}{1}
    \C_z := \W_z\begin{bmatrix} 1 & 0 & -1& 0\\ 
        0& 1 & 0 & 0
    \end{bmatrix},
    \D_d := \begin{bmatrix}0\\0\end{bmatrix},
\end{alignat*}
where $\W_z:=\diag{\begin{bmatrix} 11.46 & 0.1 \end{bmatrix}}$, scales $\z$ such that the desired closed-loop performance, defined by the $\mathcal{H}_2$ or $\mathcal{H}_\infty$ norm of the transfer function from $\begin{bmatrix}\dbar\\ \wbar_a\end{bmatrix} \to \z$, is bounded by one.

Results from theorems 1 and 2, will quantify the maximum actuator degradation in thrust ($T$), elevator ($\delta_e$), and leading-edge flap ($\delta_{lef}$). 

\subsection{Simulation Results}
The numerical values for the minimum actuator cut-off frequencies and dc gains, and maximum actuator noises are shown in table \ref{tab1}, and are plotted in \fig{act_cutoff}, \fig{act_mag}, and \fig{act_noise}.
\def\arraystretch{1.5}
\renewcommand{\tabcolsep}{0.15cm}
%% HD: UPDATING THE TABLE (Progress: DONE)
\begin{table}[h!]\centering 
    \begin{tabular}{c|rr|rr|rr}\hline
         & \multicolumn{2}{c|}{$\vo{\omega}^\ast_c$ (rad/s)} & \multicolumn{2}{c|}{$\|\x\to x_{F_i}\|^\ast_\infty$} & \multicolumn{2}{c}{$1/\sqrt{\kappa_i^\ast}$} \\ \hline
         Actuator & $\mathcal{H}_2$ & $\mathcal{H}_\infty$ & $\mathcal{H}_2$ & $\mathcal{H}_\infty$ & $\mathcal{H}_2$ & $\mathcal{H}_\infty$\\ \hline
         $T$ & 0.0033 &2.4381 &2.4945 &0.0207 &0.7157 &0.0402 \\
         $\delta_e$ & 12.9757 &16079.2678 &16.1579 &0.1707 &0.0246 &0.0014\\
         $\delta_{lef}$ & 2.0541 &2033.7134 &17.8508 &0.4800 &0.0468 &0.0027\\ \hline
    \end{tabular}\vspace{2mm}
    \caption{Parameters that quantify maximum actuator degradation for a given closed-loop performance ($\gamma$ = 0.5).} \label{tab1}
\end{table}
as
\begin{figure}[h!]
    \includegraphics*[width=0.4\textwidth]{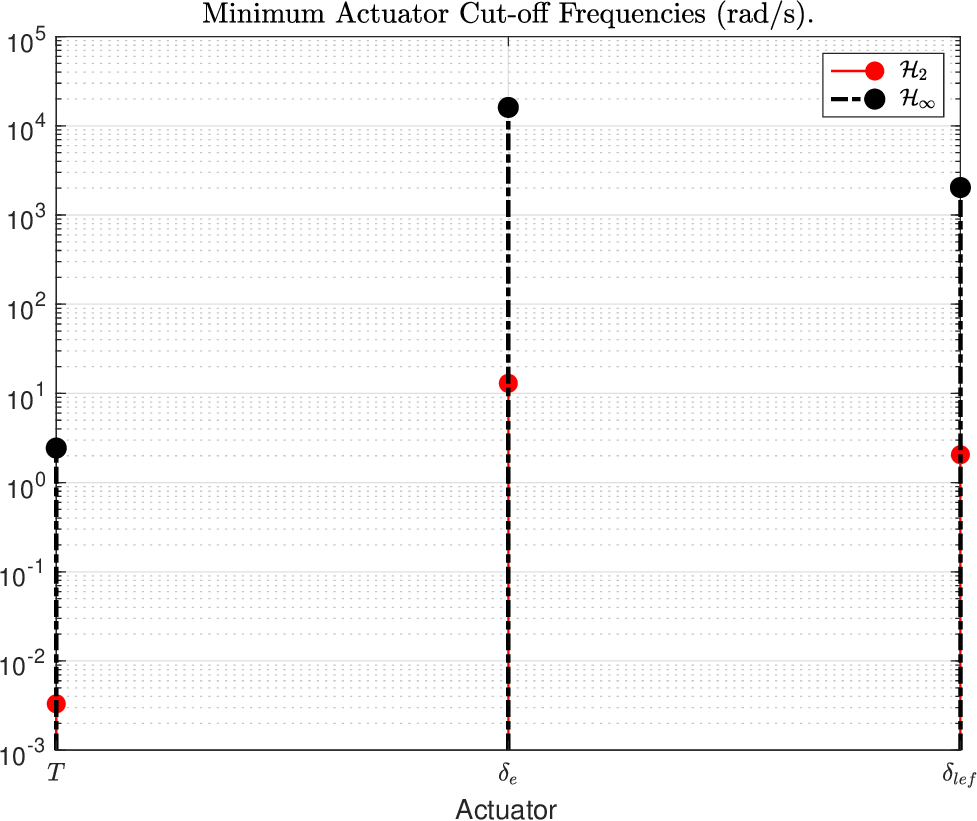}
    \caption{Minimum actuator cutoff frequencies (rad/s).}
    \flab{act_cutoff}
\end{figure}

\begin{figure}[h!]
    \includegraphics*[width=0.4\textwidth]{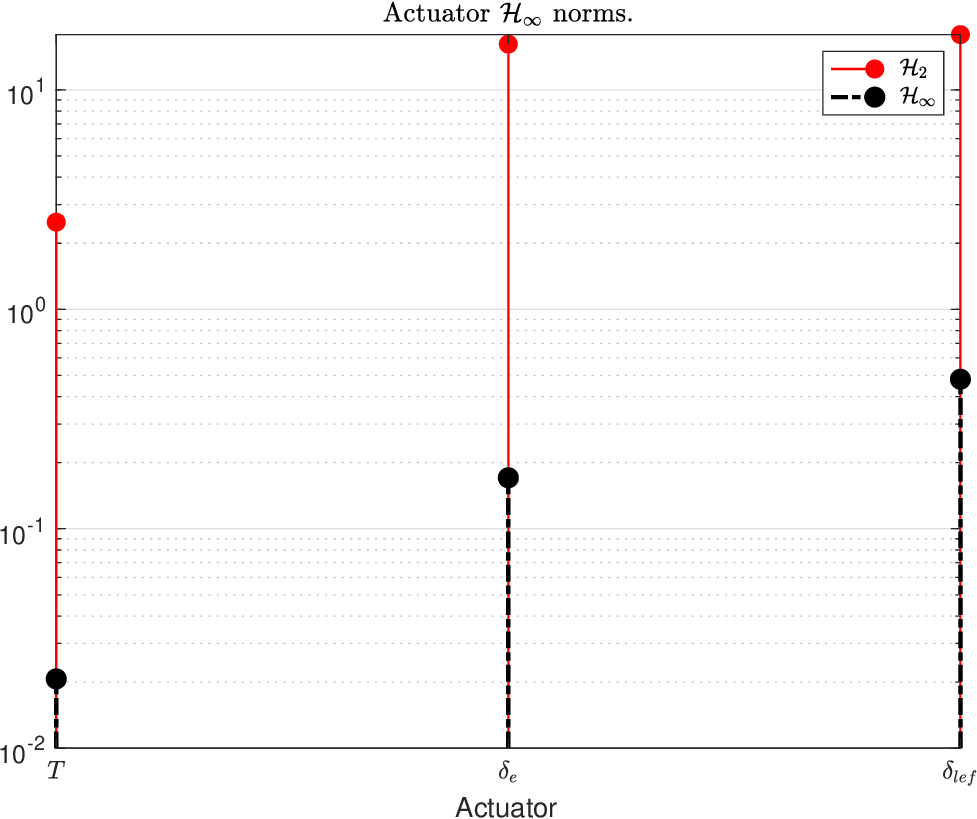}
    \caption{Minimum actuator dcgain.}
    \flab{act_mag}
\end{figure}

\begin{figure}[h!]
    \includegraphics*[width=0.4\textwidth]{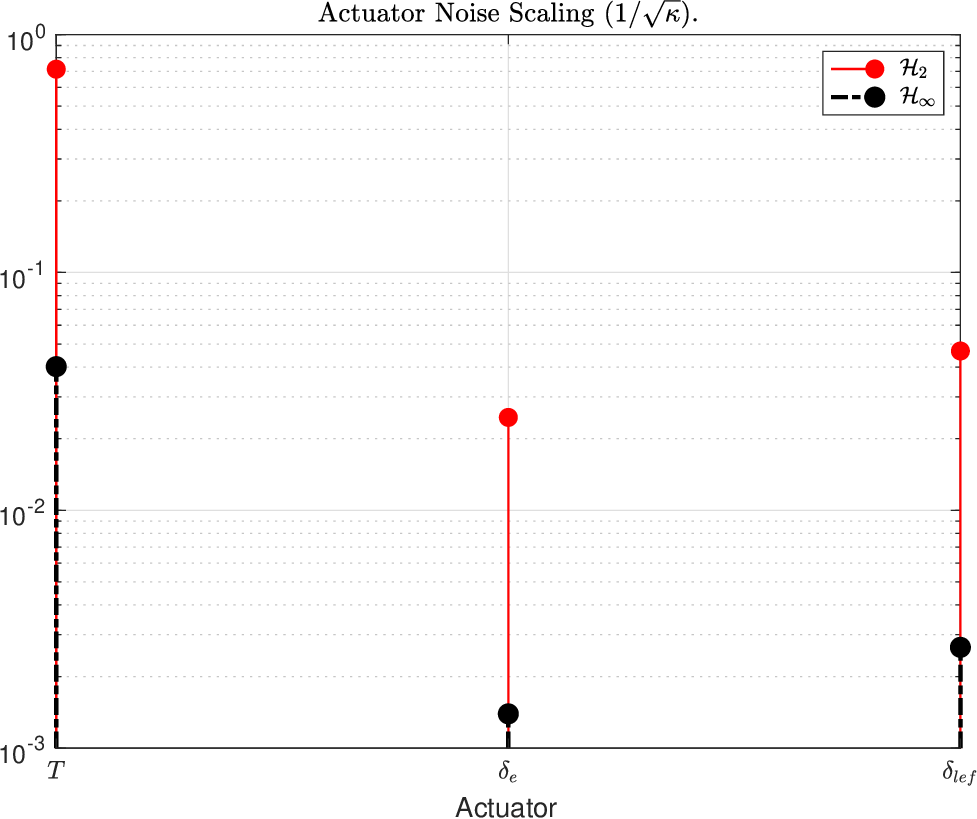}
    \caption{Maximum actuator noise scaling.}
    \flab{act_noise}
\end{figure}

\begin{figure}[h!]
    \includegraphics*[width=0.5\textwidth]{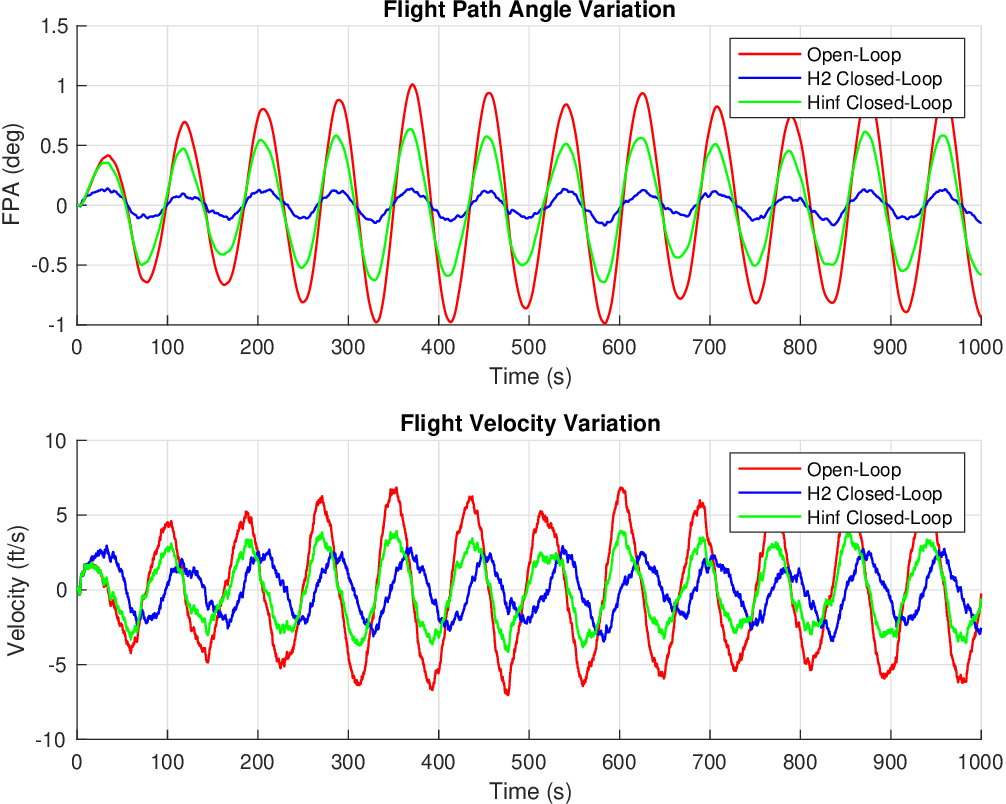}
    \caption{Time Response to Gust, $d(t) = 15\text{wn}(t) + \sin(0.075t)$, where $\text{wn}(t)$ is Gaussian white noise.}
    \flab{trsp}
\end{figure}

Table 1 indicates that the control rate for thrust is sigificantly lower than that for the elevator and leading-edge flap, with the elevator requiring a higher actuation rate compared to the leading-edge flap. Furthermore, the actuation rate resulting from the $\mathcal{H}_2$ formulation is substantially lower than that from the $\mathcal{H}_\infty$ formulation. This discrepancy arises because the exogenous inputs are presumed to be broad-band, and the $\mathcal{H}_\infty$ formulation attenuates disturbances across this broad range. Employing a frequency-weighted model for the exogenous signal is anticipated to yield improved results and will be explored in our future work.

The control magnitude, which is quantified by the $\mathcal{H}_\infty$ norm of the actuator dynamics, is significantly higher for the $\mathcal{H}_2$ closed-loop performance than the $\mathcal{H}_\infty$ closed-loop performance. For the $\mathcal{H}_2$ case, the elevator's and the leading-edge-flap's magnitude are comparable and significantly higher than thrust magnitude. For the $\mathcal{H}_\infty$, we see similar trend, but orders of magnitude smaller. 

The final two columns of Table 1 present the noise scaling for actuator disturbances. It is noted that the noise levels for the elevator and the leading-edge flap are relatively low, indicating that these actuators require high precision to achieve the specified closed-loop performance. Additionally, the noise levels for the $\mathcal{H}_\infty$ case are substantially lower than those for the $\mathcal{H}_2$ case, a disparity attributed to considering broad-band disturbances in the problem.

Based on these observations, we see that the $\mathcal{H}_2$ performance requires lower control rate and higher control magnitude with less precisions, and the $\mathcal{H}_\infty$ performance requires much higher control rate and lower control magnitude with high precisions. 

The time responses to the disturbance for both the open-loop system and the closed-loop system with actuator degradations are shown in \fig{trsp}. It is observed that the $\mathcal{H}_2$ optimal controller exhibits superior disturbance attenuation compared to the $\mathcal{H}_\infty$ controller. This outcome aligns with expectations, as $\norm{z}_\infty$ is bounded in the former case, while $\norm{z}_2$ is bounded in the latter. %The performance of the $\mathcal{H}_\infty$ controller, in contrast, does not significantly improve upon the open-loop system. This relatively conservative performance is attributed to the broadband attenuation inherent in the $\mathcal{H}_\infty$ design. We will address this convervatism in our future work.

\section{Summary \& Conclusions}
In this paper, we investigate the quantification of maximum actuator degradation in linear dynamical systems for a given closed-loop performance. We introduce a novel unified framework for computing the state-feedback controller gain, which meets a user-specified closed-loop performance criterion while maximizing actuator degradation. This degradation is modeled as a first-order filter with additive noise. We present two new convex optimization formulations that simultaneously determine the controller gain, maximize actuator degradation, and ensure the desired closed-loop performance in both the $\mathcal{H}_2$ and $\mathcal{H}_\infty$ system norms. The results are limited to open-loop stable systems. 

We apply our findings to the design of a full-state feedback controller for the longitudinal motion of the F-16 aircraft. Our observations indicate that the $\mathcal{H}_2$ performance requires a lower control rate and higher control magnitude with less precision, while the $\mathcal{H}_\infty$ performance demands a higher control rate and lower control magnitude with high precision. The $\mathcal{H}_\infty$ outcomes are conservative due to the broad-band disturbance models. We aim to extend this to frequency-weighted input-output formulations in future work.
\bibliographystyle{unsrt}
\bibliography{references}
\end{document}